\documentclass[journal,final,letterpaper]{IEEEtran}
\usepackage{amsmath,amssymb,graphicx, subfigure, epsfig}
\usepackage{graphicx}
\usepackage{tikz}
\usetikzlibrary{matrix,positioning}
\usepackage{filecontents}
\usepackage{algorithm2e}
\usetikzlibrary{snakes}
\usetikzlibrary{arrows,decorations,backgrounds}
\newcommand{\LB}{\left(}
\newcommand{\RB}{\right)}
\newcommand{\LSB}{\left[}
\newcommand{\RSB}{\right]}
\newcommand{\LCB}{\left\{}
\newcommand{\RCB}{\right\}}
\newcommand{\lv}{\lvert}
\newcommand{\rv}{\rvert}

\newcommand{\htp}{^{\sf H}}
\newcommand{\tp}{^{\sf T}}

\newcommand{\EE}{{\mathbb E}}
\newcommand{\PP}{{\mathbb P}}

\newcommand{\f}{{\bf f}}
\newcommand{\g}{{\bf g}}
\newcommand{\h}{{\bf h}}

\newcommand{\w}{{\bf w}}
\newcommand{\x}{{\bf x}}

\newcommand{\zerov}{{\bf 0}}

\newcommand{\G}{{\bf G}}

\newcommand{\I}{{\bf I}}

\newcommand{\W}{{\bf W}}
\newcommand{\X}{{\bf X}}

\newcommand{\Cc}{{\cal C}}

\newcommand{\Nc}{{\cal N}}

\newtheorem{thm}{Theorem}
\newtheorem{lem}{Lemma}
\newtheorem{definition}{Definition}

\newtheorem{rem}{Remark}

\newtheorem{app}{Approximation}

\newcommand{\mr}{\mathrm}
\usepackage{color}
\usepackage{acronym}
\usepackage{hyperref}
\usepackage{cite}
\begin{document}

\title{Large-Scale MIMO versus Network MIMO for Multicell Interference Mitigation}

\author{Kianoush Hosseini,~\IEEEmembership{Student Member,~IEEE,}
Wei Yu,~\IEEEmembership{Fellow Member,~IEEE,} \\
and Raviraj S. Adve,~\IEEEmembership{Senior Member,~IEEE}
\thanks{Manuscript received September 29, 2013; revised January 30, 2014, April
01, 2014; accepted May 12, 2014. Date of publication May 30, 2014. This work is supported by Natural Science and Engineering Research Council (NSERC) of Canada. This paper has been presented in part at IEEE International Workshop on Signal Processing Advances in Wireless Communications (SPAWC), Toronto, ON, Canada, June 2014.

Authors are with The Edward S. Rogers Sr. Department of Electrical and Computer Engineering, University of Toronto, 10 King's College Road, Toronto, Ontario, M5S 3G4, Canada (e-mails: kianoush,~weiyu,~rsadve@comm.utoronto.ca).}}

\maketitle

\begin{abstract}
This paper compares two important downlink multicell interference mitigation techniques, namely, large-scale (LS) multiple-input multiple-output (MIMO) and network MIMO. We consider a cooperative wireless cellular system operating in time-division duplex (TDD) mode, wherein each cooperating cluster includes $B$ base-stations (BSs), each equipped with multiple antennas and scheduling $K$ single-antenna users. In an LS-MIMO system, each BS employs $BM$ antennas not only to serve its scheduled users, but also to null out interference caused to the other users within the cooperating cluster using zero-forcing (ZF) beamforming. In a network MIMO system, each BS is equipped with only $M$ antennas, but interference cancellation is realized by data and channel state information exchange over the backhaul links and joint downlink transmission using ZF beamforming. Both systems are able to completely eliminate intra-cluster interference and to provide the same number of spatial degrees of freedom per user. Assuming the uplink-downlink channel reciprocity provided by TDD, both systems are subject to identical channel acquisition overhead during the uplink pilot transmission stage. Further, the available sum power at each cluster is fixed and assumed to be equally distributed across the downlink beams in both systems. Building upon the channel distribution functions and using tools from stochastic ordering, this paper shows, however, that from a performance point of view, users experience better quality of service, averaged over small-scale fading, under an LS-MIMO system than a network MIMO system. Numerical simulations for a multicell network reveal that this conclusion also holds true with regularized ZF beamforming scheme. Hence, given the likely lower cost of adding excess number of antennas at each BS, LS-MIMO could be the preferred route toward interference mitigation in cellular networks.
\end{abstract}

\section{Introduction}\label{sec:intro}
\PARstart{I}{n} traditional wireless cellular networks, base-stations (BSs) operate independently and treat out-of-cell interference as additional background noise. However, as the BSs become densely deployed, intercell interference is considered as a key limiting factor in advanced wireless networks. A promising approach for intercell interference mitigation is multicell coordination that aligns the transmit and receive strategies of the BSs in order to reduce, or even to completely eliminate, intercell interference. While the coordination strategies are of crucial significance to meet quality-of-service (QoS) requirements of future wireless networks, a systematic comparative analysis of their performances is not yet available. The main objective of this paper is to compare two distinct downlink interference mitigation techniques: large-scale (LS) multiple-input multiple-output (MIMO) and network MIMO.

In an LS-MIMO system, intra-cluster interference mitigation is achieved at the expense of needing to accommodate a large number of antennas at each cell-site. In particular, with the knowledge of channel state information (CSI) to the users within a cluster, each BS exploits its excess number of spatial dimensions not only to serve its multiple associated users using downlink beamforming, but also to employ an interference coordination (IC)~\cite{GHHSSY10} scheme to choose its beam directions in order to null out the interference caused to the other users within the cluster.

In a network MIMO system, multiple scattered BSs share the CSI of the users within the cluster as well as their data symbols through backhaul links. Hence, joint downlink data transmission becomes feasible~\cite{GHHSSY10,JTSHSP08,HMKLC11,KFV06}. Specifically, by jointly designing the downlink beams to spatially multiplex multiple users, intra-cluster interference can be completely eliminated. Network MIMO systems are often hailed as capable of achieving the ultimate capacity limit of cellular networks~\cite{GHHSSY10}.

This paper shows that the interference mitigation enabled by exploiting a large number of antennas at each cell-site significantly outperforms the joint processing scheme in a network MIMO system under a general class of utility functions. Consequently, given the likely lower cost of adding extra antennas at each cell-site versus establishing data-sharing through the backhaul and joint transmission across the cooperating BSs~\cite{RC09}, LS-MIMO could be the preferred route toward a practical realization of interference mitigation in multicell wireless networks.

\subsection{Related Work}

The LS-MIMO system considered in this paper is akin to a \emph{non-cooperative} massive MIMO system~\cite{M10,RPLLM13,LTEM13,YM13}, wherein each BS is equipped with \emph{asymptotically} large number of antennas. In this regime, the effect of uncorrelated intercell interference vanishes and multicell coordination is not required. The fundamental reason that a massive MIMO system is able to completely mitigate interference is that it emulates an \emph{interference coordination} scheme as the number of antennas goes to infinity. However, this desirable feature of massive MIMO networks comes at the expense of significant additional cost for the large number of required analog front-ends. Furthermore, the antenna elements must be well-separated so that their mutual coupling can be neglected. Consequently, to accommodate the antennas, an implementation of a massive MIMO system requires significant physical dimensions. This paper, however, considers an operating regime with a \emph{finite} and \emph{fixed} number of BS antennas. Specifically, each BS is equipped with only enough antennas to cancel intra-cluster interference, and to provide a finite number of spatial degrees of freedom (DoF) per user. Such a system is not only more practical, but also by defining fixed \emph{cooperating} clusters, it makes the comparison with a network MIMO system feasible.

The benefits of BS cooperation schemes have been extensively investigated in the literature~\cite{GHHSSY10,LSCHMNS12}. In particular, IC systems with beamforming and power adaptation have been explored~\cite{DY10,YKS13,HQSH11}. To account for the heterogeneity of the transmitter locations, stochastic geometry approaches have been used to derive closed form expressions for various performance metrics with IC via zero-forcing (ZF) beamforming in ad hoc~\cite{KA12,JAW11,HAGHB12} and cellular~\cite{VKKG13} networks.

Practical system designs for network MIMO systems considering limited backhaul capacity~\cite{MF08,MF09}, and cooperation across only a relatively small set of adjacent BSs~\cite{MF08,RCP09,CRP10,HTC12,BH07} have been investigated. In a related line of work, distributed antenna systems (e.g.~\cite{K96,CA07_2}) have been compared to co-located antenna systems. However, as acquiring new cell-sites could be costly for network providers, this paper assumes that the number of cell-sites is constant, but allows the number of antennas and corresponding backhaul requirements to vary in the comparative study.

\subsection{Main Contributions}

Despite the numerous proposed interference management algorithms and analytic investigations for LS-MIMO and network MIMO techniques, their performance has been compared only in special cases, e.g. a downlink of a two-cell cellular network~\cite{ZH12}. In contrast to the aforementioned series of works, rather than focusing on one coordination strategy, the main emphasis of this paper is an analytic comparison between LS-MIMO and network MIMO techniques. In order to place the comparison on the concrete footing, this paper considers the downlink of a time-division duplex (TDD) multicell network where each cooperating cluster comprises $B$ cells. We compare the following two systems with identical BS deployment:

\begin{itemize}
\item An LS-MIMO system where each BS is equipped with $BM$ transmit antennas, and spatially multiplexes $K$ users from within its cell area via ZF beamforming with $K \leq M$, while choosing its beam directions so as not to interfere with the other $K \LB B - 1\RB$ users in its cluster by adopting an IC scheme.

\item A network MIMO system where each BS is equipped with $M$ transmit antennas, and schedules $K$ users from within its cell area with $K \leq M$. The $BK$ scheduled users are then jointly served by the cooperating BSs using ZF beamforming.	
\end{itemize}

Although the LS-MIMO and network MIMO systems described here require distinct network infrastructure, the two systems share common features. First, by invoking the orthogonality property of ZF beamforming~\cite{YG06}, both systems provide the same number of spatial DoF per user, i.e., $\zeta = B\LB M - K \RB + 1$. Even though each BS in an LS-MIMO system is equipped with $B$ times more antenna elements than in a network MIMO system, it sacrifices many of the spatial dimensions provided by the transmit antennas to cancel intra-cluster interference. Second, given that both systems serve the same number of users during each time-slot, the cost of CSI acquisition is identical as well. The reason is that, in TDD networks, the downlink CSI can be obtained from uplink pilot transmission. Hence, the estimation overhead is only a function of the number of users and is independent of the number of antennas at the BSs~\cite{M10,HBD13}. Finally, the total transmit power of each cooperating cluster is kept the same and is equally distributed across the downlink beams in both scenarios. As a consequence, it is not immediately obvious that one system is superior to the other from a performance point of view.

This paper shows that these two systems in fact have stark differences in terms of signal propagation characteristics. Note that each user is only associated with its adjacent BS in an LS-MIMO system, while it receives its intended signal transmitted from multiple scattered BSs in a network MIMO system. The key observation of this paper is that the disparity in the distances between a given user location and its set of serving BSs incurs a penalty in the received signal power in a network MIMO system. Therefore, the signal power is statistically stronger under an LS-MIMO system for any chosen user. Further, we establish the equivalence, in distribution, of the aggregate interference power experienced at each user location in the two systems. Utilizing tools from stochastic orders~\cite{SS94,R11}, we then incorporate these two results and provide a careful stochastic analysis of the signal-to-interference-plus-noise ratio (SINR) of any chosen user under both systems. Our analytical derivations demonstrate that an LS-MIMO system provides significant gains in terms of key network performance metrics, e.g., coverage probability, ergodic rate, and expected weighted sum rate as compared to a network MIMO system. 

Further, our numerical results illustrate that the conclusion of the paper also holds true if regularized ZF (RZF) beamforming is adopted in the network. This is of practical interest since, without imposing any further implementation complexity as compared to ZF beamforming, the RZF scheme provides considerable performance gains for both systems.

Finally, it is worth mentioning that the stochastic ordering techniques have been used to study various aspects of wireless networks; e.g., coverage probability in a non-cooperative multicell network where each BS either serves multiple users via ZF beamforming or employs single-user beamforming~\cite{DKA13}, and interference ordering under different fading and node location distributions~\cite{LT12}.

\subsection{Paper Organization}
The remainder of the paper is organized as follows. Section~\ref{sec:example} provides an illustrative example of the comparison carried out in this paper. Section~\ref{sec:sys_signal_model} summarizes the system assumptions and received signal models under both architectures. Section~\ref{sec:sig_dist} presents the distribution functions of the signal and interference powers in the two systems. Section~\ref{sec:epa} establishes the stochastic ordering results and further evaluates the performance of the two systems. Section~\ref{sec:numerical} presents the numerical results. Finally, Section~\ref{sec:conc} concludes the paper.

\subsection{Notation}
In this paper, matrices, column vectors, and scalars are presented by bold capital letters $\X$, bold letters $\x$, and lowercase letters $x$, respectively. The transpose and Hermitian transpose are, respectively, denoted by $\LB \cdot \RB \tp$ and $\LB \cdot \RB \htp$. The $N \times N$ identity matrix is represented by $\I_N$. A complex Gaussian distribution function with mean $\zerov$ and covariance matrix $\X$ is given by $\Cc \Nc \LB \zerov,\X \RB$. Moreover, $\Gamma \LB k, \theta \RB$ denotes the Gamma distribution function with the shape parameter $k$ and the scale parameter $\theta$. $\PP_{X} \LB A \RB$ denotes the probability of event $A$ and $\EE_{X} \LB \cdot \RB$ is the expectation operation with respect to a random variable $X$. Finally, equivalence in distribution is denoted by $\stackrel{ d }{=}$.

\section{An Illustrative Example}\label{sec:example}

Consider a wireless network consisting of two cooperating BSs each scheduling $K = 1$ user during a given time-slot as schematically depicted in Fig.~\ref{fig:example}. The BSs are either equipped with $N_t = 2$ transmit antennas and sacrifice one spatial DoF to null out the interference on the neighboring user or are equipped with $N_t = 1$ antenna and form a network MIMO system. Note that both systems have enough antenna elements to completely null out intercell interference and to provide one spatial DoF per user. The channel vectors capture both small-scale fading and path-loss.

In the LS-MIMO system, user $1$ is being served using the two co-located antennas installed at its closest cell-site; on the other hand, in the network MIMO system, one of its serving antennas is at BS $2$, located further away. The key observation made in this paper is the following: although both $g_{11}^{\mr{LSM}}$ and $g_{11}^{\mr{NM}}$ have identical statistics, $g_{12}^{\mr{NM}}$ is statistically weaker than $g_{12}^{\mr{LSM}}$. Therefore, the received signal power in a network MIMO system is expected to be statistically weaker than that of a comparable LS-MIMO system. One of the objectives of this paper is to make this observation rigorous by a careful statistical comparison of the signal power at an arbitrarily chosen user under the two cooperation schemes.

Although both systems completely eliminate intra-cluster interference, they choose their ZF beam directions according to distinct channel matrices. Therefore, one expects different inter-cluster interference patterns to be created by the two systems. However, it is notable that a user located outside the service area of the cooperating group shown in Fig.~\ref{fig:example} receives \emph{two interfering beams} transmitted from \emph{the same set of BS locations} in both systems. Although not so obvious, the second goal of this paper is to show that, with the same number of interfering beams and identical BS deployment, the interference powers produced by the two systems at any user location are equal in distribution.

Clearly, an analysis to achieve these two objectives needs to take into account the propagation characteristics, i.e., path-loss and small-scale fading, in each system.

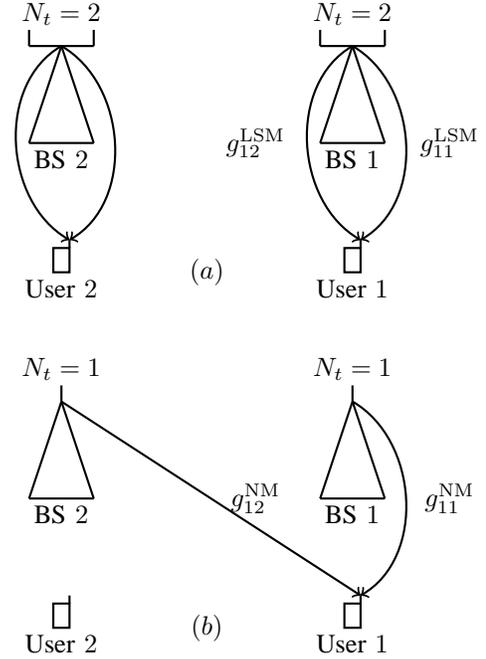
\begin{figure}[t!]
\centering
\begin{tikzpicture}[scale = 0.43]

\draw [thick] (-1,1) -- (1,1);
\draw [thick] (-1,1) -- (0,4);
\draw [thick] (1,1) -- (0,4);
\node (v1) at (0,0.5) [black] {BS $2$};
\draw [thick] (-1,4)--(1,4);
\draw[thick] (-1,4)--(-1,4.5);
\draw [thick] (1,4)--(1,4.5);
\node (vv) at (0,5) [black] {$N_t = 2$};
\draw [thick] (-0.25,-3) rectangle (0.25,-2.25);
\draw [thick] (0.25,-2.25)--(0.25,-2);
\node (vv) at (0,-3.5) [black] {User $2$};
\draw[->,thick] (0,4) to [out = -30, in = 30] (0.25,-2);
\draw[->,thick] (0,4) to [out = 210, in = 150] (0.25,-2);

\draw [thick] (8,1) -- (10,1);
\draw [thick] (8,1) -- (9,4);
\draw [thick] (10,1) -- (9,4);
\node (v1) at (9,0.5) [black] {BS $1$};
\draw [thick] (8,4)--(10,4);
\draw[thick] (8,4)--(8,4.5);
\draw [thick] (10,4)--(10,4.5);
\node (vv) at (9,5) [black] {$N_t = 2$};
\draw [thick] (8.75,-3) rectangle (9.25,-2.25);
\draw [thick] (9.25,-2.25)--(9.25,-2);
\node (vv) at (9,-3.5) [black] {User $1$};

\node at (4.5,-3) [black] {$\LB a \RB$};

\draw[->,thick] (9,4) to [out = -30, in = 30] (9.25,-2);
\draw[->,thick] (9,4) to [out = 210, in = 150] (9.25,-2);
\node (vv) at (12,1) [black] {$g_{11}^{\mr{LSM}}$};
\node (vv) at (6,1) [black] {$g_{12}^{\mr{LSM}}$};
\draw [thick] (-1,-10) -- (1,-10);
\draw [thick] (-1,-10) -- (0,-7);
\draw [thick] (1,-10) -- (0,-7);
\node (v1) at (0,-10.5) [black] {BS $2$};
\draw [thick] (0,-7)--(0,-6.5);
\node (vv) at (0,-6) [black] {$N_t = 1$};
\draw [thick] (-0.25,-14) rectangle (0.25,-13.25);
\draw [thick] (0.25,-13.25)--(0.25,-13);
\node (vv) at (0,-14.5) [black] {User $2$};

\draw [thick] (8,-10) -- (10,-10);
\draw [thick] (8,-10) -- (9,-7);
\draw [thick] (10,-10) -- (9,-7);
\node (v1) at (9,-10.5) [black] {BS $1$};
\draw [thick] (9,-7)--(9,-6.5);
\node (vv) at (9,-6) [black] {$N_t = 1$};

\draw [thick] (8.75,-14) rectangle (9.25,-13.25);
\draw [thick] (9.25,-13.25)--(9.25,-13);
\node (vv) at (9,-14.5) [black] {User $1$};

\node at (4.5,-14) [black] {$\LB b \RB$};

\draw[->,thick] (9,-7) to [out = -30, in = 30] (9.25,-13);
\draw[->,thick] (0,-7) to  (9.25,-13);
\node (vv) at (12,-10) [black] {$g_{11}^{\mr{NM}}$};
\node (vv) at (6,-10) [black] {$g_{12}^{\mr{NM}}$};
 \end{tikzpicture}
 \caption{A multicell wireless network where $B = 2$ BSs form $\LB a \RB$ an LS-MIMO system $\LB b \RB$ a network MIMO system.}
 \label{fig:example}
\end{figure}

\section{System Model}\label{sec:sys_signal_model}

We consider a multicell multiuser wireless cellular network comprising a set of $C$ disjoint cooperating clusters of BSs formed by a fixed lattice.\footnote{For simplicity, we assume that clusters are formed using a square lattice in our numerical simulations. However, any clustering approach that forms a partition of the network coverage area can be applied.} Each cluster includes $B$ cooperating BSs, each equipped with multiple transmit antennas, and schedules $K$ single-antenna users from a set of potential users scattered within its cell. Therefore, during each time-slot, a total of $K_{c} = BK$ users are being served in each individual cluster. BSs are constrained to have maximum available power $P_T$ and transmit concurrently over a shared spectrum of bandwidth $W$ with universal frequency reuse.


The channel vector between BS $m$ in cluster $j$ and user $i$ in cluster $l$ is denoted as $\sqrt{\beta_{ilmj}}\h_{ilmj}$, where $\h_{ilmj}$ defines the small-scale Rayleigh channel fading with $\Cc \Nc \LB 0, 1\RB$ independent and identically distributed (i.i.d.) components, and $\beta_{ilmj}$ denotes the distance-dependent path-loss coefficient. We consider a standard path-loss model $\beta_{ilmj} = r_{ilmj}^{-\alpha}$, where $r_{ilmj}$ denotes the distance between BS $m$ in cluster $j$ and user $i$ in cluster $l$, and $\alpha$ respresents the path-loss exponent.
Further, since the channel estimation overhead does not influence the comparison of this paper, perfect channel estimation is assumed.\footnote{The cooperating BSs are required to share their CSI through backhaul links in a network MIMO system. The conclusion of this paper is valid even without accounting for this signaling overhead.}

For analytical tractability, we further impose the following simplifying assumptions:
\begin{enumerate}
\item[$A1:$] ZF beamforming is adopted in this paper. Even though suboptimal in general, ZF beamforming achieves the same asymptotic sum rate as that of the non-linear dirty-paper coding as the number of users increases, while being significantly less complex~\cite{YG06}.

\item[$A2:$] The ZF beams designed in each cluster are, in general, not orthogonal. However, in order to characterize the interference power distributions in the two architectures, similar to other related works,  e.g.~\cite{HWKS11,DKA13,CKA09,SHCS13}, we treat the ZF beams as orthogonal vectors. Our numerical results, however, confirm the accuracy of this approximation.


\item[$A3:$] We assume that each cluster is subject to a sum power constraint. Although a per-BS power constraint is more practical, designing downlink ZF beams with per-BS power constraint is computationally complex~\cite{Z10}.

\item[$A4:$] Further, the total power of each cluster is equally distributed across the downlink ZF beams. Although power optimization can further improve system performance, it also requires significant intra-cluster and inter-cluster coordination.

\item[$A5:$] Finally, we consider round-robin scheduling. As a consequence, both coordination schemes serve the same set of $K_c$ users during each time-slot.
\end{enumerate}

In this paper, a wireless cellular system as described above is denoted as a $\LB C, B, N_t, K_c\RB$ system, wherein $N_t$ denotes the number of BS antennas. Further, throughout the paper, subscript $i$ refers to the same user in the two systems, which is served by BS $b$ in cluster $l$ in an LS-MIMO system and is jointly served by the BSs in cluster $l$ in a network MIMO system.

The rest of this section is devoted to presenting the received signal models and SINR expressions in both systems. 

\subsection{Received Signal Model in an LS-MIMO System} \label{sec:sig_model_ic}
Let user $i$ be scheduled by BS $b$ in cluster $l$ during time-slot $t$. This BS is assumed to have perfect knowledge of its CSI to the $K_c$ users within its cluster, i.e., the compound channel matrix $\G_{bl} = \LSB \g_{1lbl}, \ldots, \g_{K_clbl}\RSB \in \Cc^{BM \times K_c}$, where $\g_{ilbl} = \sqrt{\beta_{ilbl}}\h_{ilbl}$. In order to spatially multiplex its $K$ users, while suppressing interference on others, BS $b$ designs its downlink ZF beamforming matrix as
\begin{align}
\tilde{\W}_{bl} &= \LSB \G_{bl} \LB \G_{bl} \htp \G_{bl}  \RB^{-1} \RSB_{1:K} = \LSB \tilde{\w}_{1bl}, \ldots, \tilde{\w}_{Kbl} \RSB\nonumber
\end{align}
where $\tilde{\w}_{ibl} \in \Cc^{BM}$ denotes the ZF beam associated with user $i$, and $\LSB \cdot\RSB_{1:K}$ chooses the $K$ columns corresponding with the $K$ scheduled users by BS $b$. Each beam is further normalized to ensure equal power assignment, i.e., $\w_{ibl} = \frac{\tilde{\w}_{ibl}}{\lv \lv \tilde{\w}_{ibl} \rv \rv}$.

Moreover, we define $\f_{ilmj} = \sqrt{\beta_{ilmj}}\h_{ilmj}$ as the interference channel between BS $m$ in cluster $j$ and user $i$ in cluster $l$. The received signal at user $i$ is therefore a sum of its intended signal transmitted by BS $b$, the inter-cluster interference, the receiver noise, and is given by
\begin{multline}
y_{ibl} = \underbrace{\sqrt{\frac{P_T}{K}}\g_{ilbl} \htp \w_{ibl} s_{ibl}}_{\text{intended signal}} + \\
 \underbrace{\sum_{j \neq l} \sum_{m = 1 }^B \sum_{k = 1}^{K} \sqrt{\frac{P_T}{K}} \f_{ilmj} \htp \w_{kmj} s_{kmj}}_{\text{inter-cluster interference}} + \underbrace{n_{ibl}}_{\text{noise}} \label{eq:sig_model_ic}
\end{multline}
where the information signal intended for user $i$ during time-slot $t$ is denoted by a complex scalar $s_{ibl}$ with unit average power, i.e.,  $\EE \LSB \lv s_{ibl} \rv^2 \RSB = 1$, and $n_{ibl}$ denotes the circularly symmetric complex additive white Gaussian noise with variance $\sigma^2$.
Therefore, the SINR of user $i$ associated with BS $b$ in cluster $l$ of an LS-MIMO system is given by
\begin{align}
\gamma_{ibl}^{\mr{LSM,ZF}} &= \frac{\rho \lvert \g_{ilbl}\htp \w_{ibl}\rvert^2}{ \sum_{ j \neq l} \sum_{m} \sum_{k} \rho \lvert \f_{ilmj} \htp \w_{kmj} \rvert^2+ 1} \label{eq:sinr_ic}
\end{align}
where $\rho = \frac{P_T}{K\sigma^2}$ indicates the signal-to-noise ratio (SNR).

\subsection{Received Signal Model in a Network MIMO System}
The $B$ cooperating BSs in cluster $l$ of a network MIMO system are assumed to have perfect knowledge of $\G_{l} = \LSB \g_{1l},\ldots,\g_{K_cl}\RSB \in \Cc^{BM \times K_c}$ denoting the compound channel matrix from the BSs to users within the cluster. Here, $\g_{il} = \LSB \g_{il1l} \tp,\ldots,\g_{ilBl}\tp\RSB \tp\in \Cc^{BM}$ indicates the composite channel vector between the $B$ serving BSs and user $i$. The downlink ZF beamforming matrix is jointly designed by the cooperating cells as
\begin{align}
\tilde{\W}_{l} &= \G_{l} \LB \G_{l} \htp \G_{l}  \RB^{-1} = \LSB \tilde{\w}_{1l},\ldots,\tilde{\w}_{K_cl}\RSB \nonumber
\end{align}
wherein $\tilde{\w}_{il} \in \Cc^{BM}$ is the beam assigned for user $i$. Further, we normalize each beam so that $\w_{il} = \frac{\tilde{\w}_{il}}{\lv \lv \tilde{\w}_{il} \rv \rv} $ to ensure equal power assignment. We define  $\f_{ilj} = \LSB \g_{il1j} \tp,\ldots, \g_{ilBj}\tp \RSB \tp$ as the composite interference channel from the BSs in cluster $j$ to user $i$ in cluster $l$.
Therefore, the received signal at user $i$ in cluster $l$ is a sum of the intended signal jointly transmitted from a set of cooperating BSs, the inter-cluster interference, the receiver noise, and is given by
\begin{align}
y_{il} &= \underbrace{\sqrt{\frac{P_T}{K}} \g_{il} \htp \w_{il} s_{il}}_{\text{intended signal}}  + \underbrace{\sum_{j \neq l} \sum_{k = 1}^{K_c} \sqrt{\frac{P_T}{K}} \f_{ilj}\htp \w_{kj} s_{kj}}_{\text{inter-cluster interference}} + \underbrace{n_{il}}_{\text{noise}}\label{eq:sig_model_nm}
\end{align}
where $s_{il}$ is a complex scalar representing the information signal for user $i$ in cluster $l$ with $\EE \LSB \lv s_{il} \rv^2\RSB = 1$, and $n_{il} $ denotes the circularly symmetric complex additive white Gaussian noise with variance $\sigma^2$. 
Hence, the SINR of user $i$ in cluster $l$ of a network MIMO system is given by
\begin{align}
\mr{\gamma}_{il}^{\mr{NM,ZF}} &= \frac{\rho \lv \g_{il}\htp \w_{il}\rv^2}{\sum_{ j \neq l} \sum_{k} \rho \lv \f_{ilj}\htp \w_{kj}\rv^2  + 1}. \label{eq:sinr_nm1}
\end{align}

\begin{rem}
As it is evident from~\eqref{eq:sinr_ic} and~\eqref{eq:sinr_nm1}, the SINR expressions involve the power of the channel vectors projected onto the beamforming subspace. Therefore, obtaining the distribution functions of the relevant terms is essential.
\end{rem}


\section{Signal and Interference Power Distributions}\label{sec:sig_dist}
This section presents the distribution functions of both the signal power and the interference power produced by transmission of a single beam in an interfering cluster under both interference mitigation techniques.

When channel vectors are isotropic, i.e, comprising i.i.d.~components, adopting ZF beamforming leads to a tractable characterization of the distribution functions associated with the signal and interference powers in terms of Gamma random variables. For a detailed discussion on the significance of isotropic assumption and ZF beamforming for deriving these distribution functions, refer to~\cite[Theorem 1.1]{FM90} and~\cite[Proof of Theorem 1]{CKA09}. Although the isotropic condition holds true for channel vectors in an LS-MIMO system, due to the different path-loss components involved within the composite channel vector, this condition is not met in a network MIMO system. We therefore employ an approximation technique pioneered in~\cite{SHCS13} to obtain the power distribution functions in a network MIMO system. 

\subsection{Distributions of the Signal and Interference Powers in an LS-MIMO System} \label{sec:ic}
From the channel model described in Section~\ref{sec:sig_model_ic}, it follows that $\g_{ilbl} \sim \Cc \Nc \LB \zerov,\beta_{ilbl}\I_{BM}\RB$ and $\f_{ilmj} \sim \Cc \Nc \LB \zerov,\beta_{ilmj}\I_{BM}\RB$. Therefore, $\g_{ilbl}$ and $\f_{ilmj}$ are random isotropic vectors in a $BM$-dimensional vector space whose powers are a sum of $BM$ independent exponentially distributed random variables, and distributed as
\begin{align}
&\g_{ilbl} \htp \g_{ilbl} \sim \Gamma \LB BM , \beta_{ilbl}\RB \nonumber \\
&\f_{ilmj} \htp \f_{ilmj} \sim \Gamma \LB BM , \beta_{ilmj}\RB. \nonumber
\end{align}

Based on the ZF orthogonality property, the beam vector associated with user $i$ is orthogonal to the subspace spanned by the channel vectors between BS $b$ and the other $K_c - 1$ users in cluster $l$, i.e.
\begin{equation}
\w_{ibl} \perp \text{span} \LCB \g_{klbl}\RCB_{k \neq i}. \nonumber
\end{equation}
As a consequence, the signal power $\lv \g_{ilbl} \htp \w_{ibl}\rv^2$ is the power of an isotropic $BM$-dimensional random vector projected onto a $BM - K_c + 1$ dimensional beamforming space~\cite{CKA09}. In light of this discussion, we have the following lemma.

\begin{lem}
\label{dist_ic_signal}
In a $\LB C, B, BM, K_c\RB$ LS-MIMO network, the signal power of user $i$ in cluster $l$ is distributed as
\begin{equation}
\lvert \g_{ilbl}\htp \w_{ibl}\rvert^2   \sim \Gamma \LB BM - K_c + 1,\beta_{ilbl} \RB.  \label{eq:ic_intended}
\end{equation}
\end{lem}

Moreover, the beam design in each cluster is independent of the interference channels to other clusters. Specifically, from the perspective of each user, an interfering beam lies within a one-dimensional vector space. Therefore, $\lv \f_{ilmj} \htp\w_{kmj}\rv^2$ is the power of the interfering channel vector $\f_{ilmj}$ projected onto a one-dimensional beamforming space~\cite{CKA09}. This is summarized as follows.

\begin{lem}
\label{dist_ic_intf}
In a $\LB C, B, BM, K_c\RB$ LS-MIMO network, the interference power caused by transmission of a single beam $k$ from BS $m$ in cluster $j$, when seen by user $i$ in cluster $l$, is distributed as
\begin{equation}
\lvert \f_{ilmj}\htp \w_{kmj}\rvert^2  \sim \Gamma \LB 1, \beta_{ilmj} \RB. \label{eq:ic_intf}
\end{equation}
\end{lem}

\subsection{Distributions of the Signal and Interference Powers in a Network MIMO System}\label{sec:nm}
Unlike an LS-MIMO system, the channel vectors do not consist of i.i.d. components in a network MIMO system. Specifically, the channel strengths
\begin{align}
\g_{il}\htp\g_{il} &= \sum_{b = 1}^B \beta_{ilbl}\h_{ilbl}\htp\h_{ilbl} \label{eq:chan_st_nm}\\
\f_{ilj}\htp\f_{ilj} &= \sum_{m = 1}^B \beta_{ilmj}\h_{ilmj}\htp\h_{ilmj} \label{eq:intf_st_nm} 
\end{align}
are summations of $B$ independent but non-identically distributed terms; the $b^{\mr{th}}$ term in~\eqref{eq:chan_st_nm} is distributed as $\Gamma \LB M, \beta_{ilbl}\RB$ and the $m^{\mr{th}}$ term in~\eqref{eq:intf_st_nm} is distributed as  $\Gamma \LB M, \beta_{ilmj}\RB$. However, the exact distribution of the sum of independent and non-identically distributed Gamma random variables does not yield a mathematically tractable expression. We therefore employ the second-order matching technique~\cite{HWKS11} and its consequence to obtain approximate distributions.

\begin{lem}
\label{Sum_Gamma}
Let $\LCB X_i\RCB_{i = 1}^{m}$ be a set of $m$ independent random variables such that $X_i \sim \Gamma \LB k_i,\theta_i \RB$. Then,
$Y = \sum_{i} X_i$ has the same first and second order statistics as a Gamma random variable with the shape and scale parameters given as
\begin{equation}
k = \frac{\LB \sum_i k_i \theta_i \RB^2}{\sum_i k_i \theta_i^2},~\theta = \frac{\sum_i k_i \theta_i^2}{\sum_i k_i \theta_i }.\label{eq:shape_scale}
\end{equation}
\end{lem}

\begin{app}
\label{app_1}
As a consequence of Lemma~\ref{Sum_Gamma}, a sum of $m$ non-identically distributed Gamma random variables where the $i^{\mr{th}}$ random variable is distributed as  $X_i \sim \Gamma \LB k_i,\theta_i \RB$ can be approximated as the Gamma random variable with 
the shape and the scale parameters as given in~\eqref{eq:shape_scale}.
\end{app}

According to Approximation~\ref{app_1}, the distributions of the channel strength $\g_{il}\htp \g_{il}$ and the interference channel strength $\f_{ilj} \htp \f_{ilj}$ can be presented, respectively, as the $\Gamma \LB k_{il},\theta_{il}\RB$ distribution and the $\Gamma \LB k_{ij},\theta_{ij}\RB$ distribution wherein 
\begin{align}
k_{il} &= M \frac{\LB \sum_{b = 1}^B \beta_{ilbl} \RB^2}{ \sum_{b = 1}^B \beta_{ilbl}^2},~\theta_{il} = \frac{\sum_{b = 1}^B \beta_{ilbl}^2} {\sum_{b = 1}^B \beta_{ilbl}} \label{eq:nm_channel} \\
 k_{ij} &= M \frac{\LB \sum_{m = 1}^B \beta_{ilmj} \RB^2}{ \sum_{m = 1}^B \beta_{ilmj}^2},~\theta_{ij} = \frac{\sum_{m = 1}^B \beta_{ilmj}^2} {\sum_{m = 1}^B \beta_{ilmj}}. \label{eq:nm_intf_dist_para}
\end{align}

From~\eqref{eq:nm_channel} and~\eqref{eq:nm_intf_dist_para}, it is easy to observe that $k_{il} \leq BM$ and $k_{ij} \leq BM$ with equality if the two vectors were isotopic. In essence, each spatial dimension (i.e. each entry) of an isotropic vector adds $1$ to the shape parameter of the power distribution function. To obtain tractable distributions in a network MIMO system,~\cite{SHCS13} proposes to treat $\g_{il}$ and $\f_{ilj}$ as isotrpic vectors while each spatial dimension only contributes a fraction to the shape parameter of the associated power distribution. This approach is presented as follows.



\begin{app}\label{app_2}
The intended channel vector $\g_{il}$ can be treated as an isotropic vector distributed according to $\Cc\Nc \LB \zerov,\theta_{il}\I_{BM}\RB$ where \emph{each spatial dimension} contributes $k_{il}/{BM}$ to the shape parameter of the power distribution function with $k_{il}$ and $\theta_{il}$ as defined in~\eqref{eq:nm_channel}. Likewise, $\f_{ilj}$ can be treated as an isotropic vector distributed according to $\Cc\Nc \LB \zerov,\theta_{ij} \I_{BM}\RB$, where each spatial dimension adds $k_{ij}/BM$ to the shape parameter associated with the vector power with $k_{ij}$ and $\theta_{ij}$ as defined in~\eqref{eq:nm_intf_dist_para}.
\end{app}

Under Approximations~\ref{app_1} and~\ref{app_2}, and noting that each beam lies in a $BM - K_c +1$ dimensional space, the shape parameter associated with the distribution of the signal power $\lv \g_{il} \htp \w_{il}\rv^2$ becomes $\LB BM -K_c + 1\RB\frac{k_{il}}{BM}$. In a similar fashion, since each interfering beam spans a one-dimensional space, the shape parameter associated with the distribution of the interference power $\lv \f_{ilj} \htp \w_{kj} \rv^2$ becomes $k_{ij}/BM$. The distribution functions corresponding to the signal power and interference  power due to the transmission of a single beam are formally stated in the following two lemmas.

\begin{lem}
\label{nm_sig_power_dist}
Under Approximations~\ref{app_1} and~\ref{app_2}, the signal power of user $i$ in cluster $l$ of a $\LB C,B,M,K_c\RB$ network MIMO system is distributed as
\begin{equation}
\lv \g_{il}\htp \w_{il} \rv^2 \sim \Gamma \LB \frac{k_{il}\LB BM - K_c + 1\RB}{BM}, \theta_{il}\RB \label{eq:nm_desired}
\end{equation}
with $k_{il}$ and $\theta_{il}$ as defined in~\eqref{eq:nm_channel}.
\end{lem}

\begin{lem}
\label{nm_intf_power_dist}
Under Approximations~\ref{app_1} and~\ref{app_2}, the interference power due to the transmission of beam $k$ in cluster $j$ at user $i$ in cluster $l$ of a $\LB C,B,M,K_c\RB$ network MIMO system is distributed as
\begin{equation}
\lv \f_{ilj}\htp \w_{kj}\rv^2 \sim \Gamma \LB \frac{k_{ij}}{BM},\theta_{ij}\RB \label{eq:nm_intf_dist} 
\end{equation}
with $k_{ij}$ and $\theta_{ij}$ as defined in~\eqref{eq:nm_intf_dist_para}.
\end{lem}


Based on the presented distribution functions, the following section establishes the stochastic ordering of the signal and aggregate interference powers under the two systems. These results play a central role in conducting a comparison between LS-MIMO and network MIMO systems in terms of various performance metrics.

\section{Ordering Result and Performance Evaluation} \label{sec:epa}
A concrete evaluation of performance gains of the two systems relies on a careful statistical analysis of the achievable SINRs. The SINR, however, depends on the network configuration, i.e., the locations of the set of serving BSs, interfering BSs, and the users. The remainder of this paper assumes an identical BS deployment in the two systems and separately investigates the statistical relation between the signal powers and aggregate interference powers at a fixed user location. Tools from stochastic orders then enable us to connect the two parts and obtain a complete statistical understanding of the achievable SINRs in the two systems.

\subsection{Stochastic Ordering of the Signal Powers}\label{sec:order_intended}
This section statistically orders the signal power at a given user location in the two systems. Specifically, this result makes the intuition provided by the example in Section~\ref{sec:example} analytically concrete. Consistent with~\cite{R11}, we first define first-order stochastic dominance as follows.

\begin{definition}[First-Order Stochastic Dominance]
A random variable $X_1$ is said to be first-order stochastically dominated by a random variable $X_2$, i.e., $X_2 \geq_{st} X_1$ if and only if
\begin{equation}
\PP \LB X_2 \geq x \RB \geq \PP \LB X_1 \geq x \RB, \forall x. \nonumber
\end{equation}
Essentially, for any $x$, the value of the complementary cumulative distribution function (CCDF) associated with $X_2$ should be no smaller than that of the $X_1$.
\end{definition}

Using Lemmas~\ref{dist_ic_signal} and~\ref{nm_sig_power_dist}, the following theorem establishes the signal power first-order stochastic dominance across the two systems.

\begin{thm}
\label{sig_order}
Under Approximations~\ref{app_1} and~\ref{app_2}, the signal power of each user in a $\LB C, B, BM, K_c\RB$ LS-MIMO system first-order stochastically dominates the signal power of the same user in a $\LB C, B, M, K_c\RB$ network MIMO system.
\end{thm}
\begin{IEEEproof}
See Appendix~\ref{sec:app0}.
\end{IEEEproof}

\begin{rem}
Theorem~\ref{sig_order} relies on the Approximations~\ref{app_1} and~\ref{app_2} to obtain the distribution function of the signal power in a network MIMO system. The accuracy of these approximations is confirmed through numerical simulations in Section~\ref{sec:numerical}.
\end{rem}

Theorem~\ref{sig_order} states that, unless a user is equidistant from its set of serving BSs in a network MIMO system (in which case the two systems perform similarly), an LS-MIMO system provides statistically stronger signal power. In particular, this theorem illustrates that the main disadvantage of a network MIMO system is that users receive their intended signals from multiple, scattered, BSs. This, in turn, leads to an additional penalty in terms of signal power for each user in a network MIMO system.

As the second step toward SINR stochastic ordering, the following section studies the aggregate interference power distribution under both systems.

\subsection{Distribution of the Aggregate Interference Powers}\label{sec:order_intf}
Since the two interference mitigation approaches generate their ZF beams based on different channel matrices, one expects distinct interference patterns to be created in the two systems. However, it is essential to note that a user sees the \emph{same number of interfering beams} initiated from \emph{the same set of BS locations} in both systems. Further,  as the channel vectors and therefore the ZF beams are isotropic in LS-MIMO systems and (approximately so) in network MIMO systems, their transmission directions are uniformly distributed. This section shows that these facts lead to aggregate interference powers that are equal, in distribution, at each user location under the two systems. 

Let the aggregate interference power created by an LS-MIMO system and a network MIMO system experienced by user $i$ in cluster $l$, respectively, be given as
\begin{align}
I_{il}^{\mr{LSM}} &= \sum_{j \neq l} \sum_{m = 1}^{B} \sum_{k = 1}^{K} \lv \f_{ilmj} \htp \w_{kmj}\rv^2 \label{eq:agg_intf_ic} \\
I_{il}^{\mr{NM}} &= \sum_{j \neq l} \sum_{k = 1}^{K_c} \lv \f_{ilj} \htp \w_{kj}\rv^2. \label{eq:agg_intf_nm}
\end{align}
The following theorem formally establishes the aggregate interference power distribution equivalence in the two systems.

\begin{thm}\label{epa_intf_dist}
Under Approximations~\ref{app_1} and~\ref{app_2}, the aggregate interference power produced by a $\LB C, B, BM, K_c\RB$ LS-MIMO system and a $\LB C, B, M, K_c\RB$ network MIMO system at each user location are equal in distribution.
\end{thm}
\begin{IEEEproof}
See Appendix~\ref{sec:app1}.
\end{IEEEproof}

\begin{rem}\label{intf_app}
Theorem~\ref{epa_intf_dist} relies on the following three approximations. First, a network MIMO interference channel vector is approximated by an isotropic vector as stated in Approximation~\ref{app_2}. Second, as mentioned in Assumption A$2$,  the ZF beams transmitted by each interfering BS of an LS-MIMO system and also the ZF beams designed in each interfering cluster of a network MIMO system are treated as orthogonal vectors. Based on this assumption, the $K$ interference signals produced by each individual BS of an LS-MIMO system and the $K_c$ interference signals produced by each cluster of a network MIMO system become independent. This is essential in obtaining the related distribution functions. Finally, Approximation~\ref{app_1} is used to obtain a tractable distribution function of the aggregate interference power created by each individual cluster of an LS-MIMO system. The simulation results provided in Section~\ref{sec:numerical} show that these approximations are very accurate and confirm the validity of our result.
\end{rem}

\begin{rem}
It is important to note that both Theorems~\ref{sig_order} and~\ref{epa_intf_dist} rely on equal power assignment across the downlink ZF beams. Specifically, since a fraction of the total power allocated to each beam is deterministic, equal power allocation does not affect the distribution functions of the signal and interference powers.
\end{rem}

\subsection{SINR Stochastic Ordering}
Theorems~\ref{sig_order} and~\ref{epa_intf_dist}, respectively, present the first-order stochastic dominance of the signal power and equivalence, in distribution, of the interference power at each user location across the two interference mitigation techniques. The following theorem incorporates these two results and establishes the stochastic dominance of the achievable SINRs under both systems. 

%
%

\begin{thm}\label{sinr_order}
Under Approximations~\ref{app_1} and~\ref{app_2}, the SINR of each user in a $\LB C, B, BM, K_c\RB$ LS-MIMO system first-order stochastically dominates the SINR of the same user in a $\LB C, B, M, K_c\RB$ network MIMO system.
\end{thm}
\begin{IEEEproof}
See Appendix~\ref{sec:app2}.
\end{IEEEproof}

As an immediate consequence of Theorem~\ref{sinr_order}, we can compare the two systems in terms of coverage probability. The coverage probability is defined as follows.
\begin{definition}[Coverage Probability]
The coverage probability of a user is defined as the probability that its SINR exceeds a pre-determined threshold $\eta_i$, i.e.
\begin{equation}
\PP \LB \gamma_i \geq \eta_i \RB. \nonumber
\end{equation}
This threshold is set based on the minimum required QoS for a given user and is chosen by the upper-layer mechanisms.
\end{definition}

Therefore, Theorem~\ref{sinr_order} implicitly shows that, for every choice of QoS threshold, an LS-MIMO system outperforms a comparable network MIMO system in terms of user coverage probability.

\subsection{Performance Evaluation}
In addition to the coverage probability, Theorem~\ref{sinr_order} can be used to evaluate the performance gains of the two systems under a general class of utility functions. For this purpose, investigating the stochastic dominance of functionals of random variables is required. Therefore, before proceeding to the main result of this section, we state the following lemma.

\begin{lem}
\label{func_order}
Let $X$ and $Y$ be two random variables such that $X \geq_{st} Y$. Then, for any non-decreasing function $g(\cdot)$, it follows that~\cite{R11}
\begin{equation}
 g \LB X \RB \geq_{st}   g \LB Y\RB. \nonumber
\end{equation}
\end{lem}

Let $U_{il}$ denote the utility function associated with user $i$ in cluster $l$. We further assume that $U_{il}$ is non-decreasing in SINR. The following theorem compares the two systems with respect to the achievable utility of a given user, when averaged over small-scale fading.


\begin{thm}\label{average_utility}
Under Approximations~\ref{app_1} and~\ref{app_2}, when averaged over small-scale fading, each user in a $\LB C, B, BM, K_c\RB$ LS-MIMO system achieves better utility than in a $\LB C, B, M, K_c\RB$ network MIMO system.
\end{thm}

\begin{proof}
From Theorem~\ref{sinr_order}, in conjunction with Lemma~\ref{func_order}, it follows that for any $t$
\begin{equation}
\PP \LB U_{il} \LB \gamma_{ibl}^{\mr{LSM,ZF}}\RB \geq t \RB \geq \PP \LB U_{il} \LB \gamma_{il}^{\mr{NM,ZF}}\RB  \geq t\RB. \nonumber
\end{equation}

Noting that for any non-negative random variable $X$, $\EE \LB X \RB =  \int_{t > 0} \PP \LB X > t \RB \mr{d}t$, we have that
\begin{align}
 \EE_{\h} \LCB U_{il} \LB \gamma_{ibl}^{\mr{LSM,ZF}} \RB \RCB &= \int_{t \geq 0} \PP \LB U_{il} \LB \gamma_{ibl}^{\mr{LSM,ZF}} \RB > t \RB \mr{d}t \nonumber \\
&\geq \int_{t > 0} \PP \LB U_{il} \LB \gamma_{il}^{\mr{NM,ZF}} \RB \geq t \RB \mr{d}t \nonumber \\
& = \EE_{\h} \LCB U_{il} \LB \gamma_{il}^{\mr{NM,ZF}} \RB \RCB. \nonumber
\end{align}
This completes the proof.
\end{proof}

Theorem~\ref{average_utility} holds true for a wide range of utility functions including some of the key performance metrics in wireless networks. In particular, letting $U_{il} \LB \gamma_{il} \RB = \log_2 \LB 1 + \gamma_{il} \RB$ wherein $\gamma_{il}$ denotes the SINR, it follows that an LS-MIMO system provides performance improvement in terms of user ergodic rate as compared to a network MIMO system.

Moreover, we define the expected achievable weighted sum-rate in cluster $l$ as
\begin{equation}
R_l \LB \LCB \gamma_{il}\RCB_i\RB = \EE_{\h} \LCB \sum_{i} \psi_{il} \log_2 \LB 1 +\gamma_{il} \RB \RCB,~\forall l\nonumber
\end{equation}
where $\psi_{il}$'s are non-negative constant weight factors introduced by the round-robin scheduler to prioritize users' service rates according to their application. Therefore, regardless of the interference mitigation technique employed in a network, the same set of users are selected in both systems during each given time-slot. Based on Theorem~\ref{average_utility}, it is easy to conclude that the expected weighted sum rate provided in each cluster of  a $\LB C, B, BM, K_c\RB$ LS-MIMO system is greater than in a $\LB C, B, M, K_c\RB$ network MIMO system.

\section{Numerical Validation} \label{sec:numerical}
This section presents the numerical results to support our conclusions made in the preceding sections. We further show that, with RZF beamforming, an LS-MIMO system is superior to a comparable network MIMO system as well.

We consider a multicell wireless cellular network where cooperating clusters are formed using a square lattice, each consisting of $B$ BSs located on a grid. During each time-slot, each BS schedules $K$ single-antenna users from within its cell area to be served. Further, in order to avoid boundary effects, we consider a wrap-around topology such that each individual cluster has eight neighboring clusters, and focus only on the cluster located in the center of the lattice. The available bandwidth is reused across the cooperating groups with a frequency reuse factor of one; simultaneous downlink transmissions therefore produce inter-cluster interference. One realization of such a network topology, with $C = 9$, $B = 4$, $K = 5$, and an inter-BS distance of $500$ meters, is depicted in Fig.~\ref{fig:network_topology}.
\begin{figure}[!t]
\centering
\includegraphics[scale = 0.44]{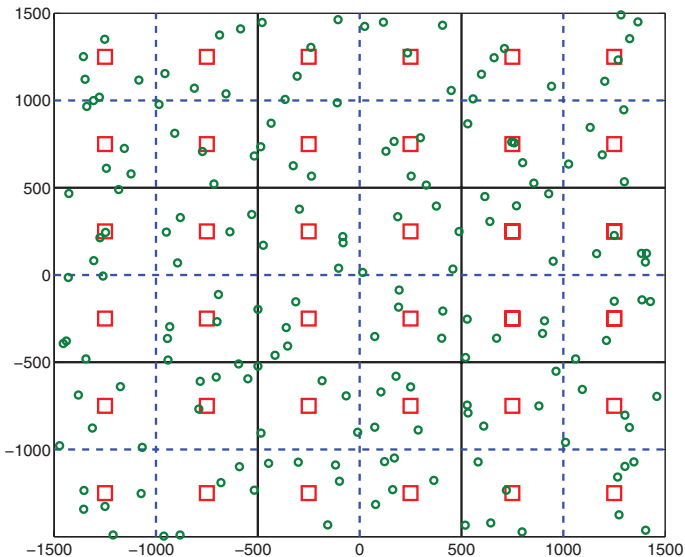}
\caption{A random snapshot of a network with $C  =9$, $B = 4$, and $K = 5$. The square and circle markers denote the cooperating BSs in each cluster and scheduled users, respectively. The solid lines show the cluster coverage area, and the dashed lines denote cell boundaries.}
\label{fig:network_topology}
\end{figure}

Each BS is assumed to be equipped with $M = 5, 6, 7$ transmit antennas in the network MIMO system which correspond, respectively, to $\zeta = 1,5,9$ per user. (Recall that $\zeta = B \LB M - K \RB + 1$ and denotes the spatial DoF per user.) In order to achieve the same number of spatial DoF per user, each BS is equipped with $BM = 20, 24, 28$ transmit antennas in the LS-MIMO system. The channel model captures the effects of the distance-dependent path-loss and Rayleigh small-scale channel fading. The network parameters are summarized in Table~\ref{table:parameters}.

\begin{table}[!t]
\caption{Network Design Parameters}
\centering
    \begin{tabular}{|c| c|}
    \hline
    Number of clusters & $C = 9$ \\ \hline
    Number of cooperating BSs per cluster & $B = 4$ \\ \hline
    Total number of users per cell & $K_T = 60$ \\ \hline
    Number of scheduled users per cell& $K = 5$ \\ \hline
    Total bandwidth & $W = 20$ MHz \\ \hline
    BS Max available power & $43$ dBm \\ \hline
    Cluster side length & $L =  1000$ m  \\ \hline
    Path-loss exponent & $\alpha = 3.5$ \\ \hline
    Background noise & $N_o = -174$ dBm/Hz \\ \hline
         \end{tabular}
    \label{table:parameters}
\end{table}

$1)$ \emph{Signal power stochastic dominance}: Fig.~\ref{fig:eff_chan_ccdf_center} and Fig.~\ref{fig:eff_chan_ccdf_corner} present the CCDF of the signal power for two user locations. First, we consider a user located at $\LB 15 \mr{m},15  \mr{m}\RB$, almost equidistant from the coordinated BSs. Second, the user is assumed to be closer to one of the BSs, located at $\LB 235 \mr{m},235 \mr{m} \RB$. While we fix the user location under consideration, the results are averaged over small-scale fading in the LS-MIMO system\footnote{Note that the ZF beams designed by each BS in an LS-MIMO system are only dependent on the small-scale channel fading between the BS and the set of scheduled users in its associated cluster.} and are averaged over the locations of the remaining users and small-scale fading realizations in the network MIMO system.

As shown in both figures and expected from Theorem~\ref{sig_order}, the CCDF of the signal powers at both user locations under the LS-MIMO system dominate those of the network MIMO system for different choices of $\zeta$. Further, as it is evident from both figures, the gap between the corresponding curves increases as the user moves closer to one of the BSs, i.e., the disparity in the distances between the user location and the set of serving BSs statistically degrades the effective signal power under the network MIMO system.

\begin{figure}[!t]
\centering
\includegraphics[scale = 0.44]{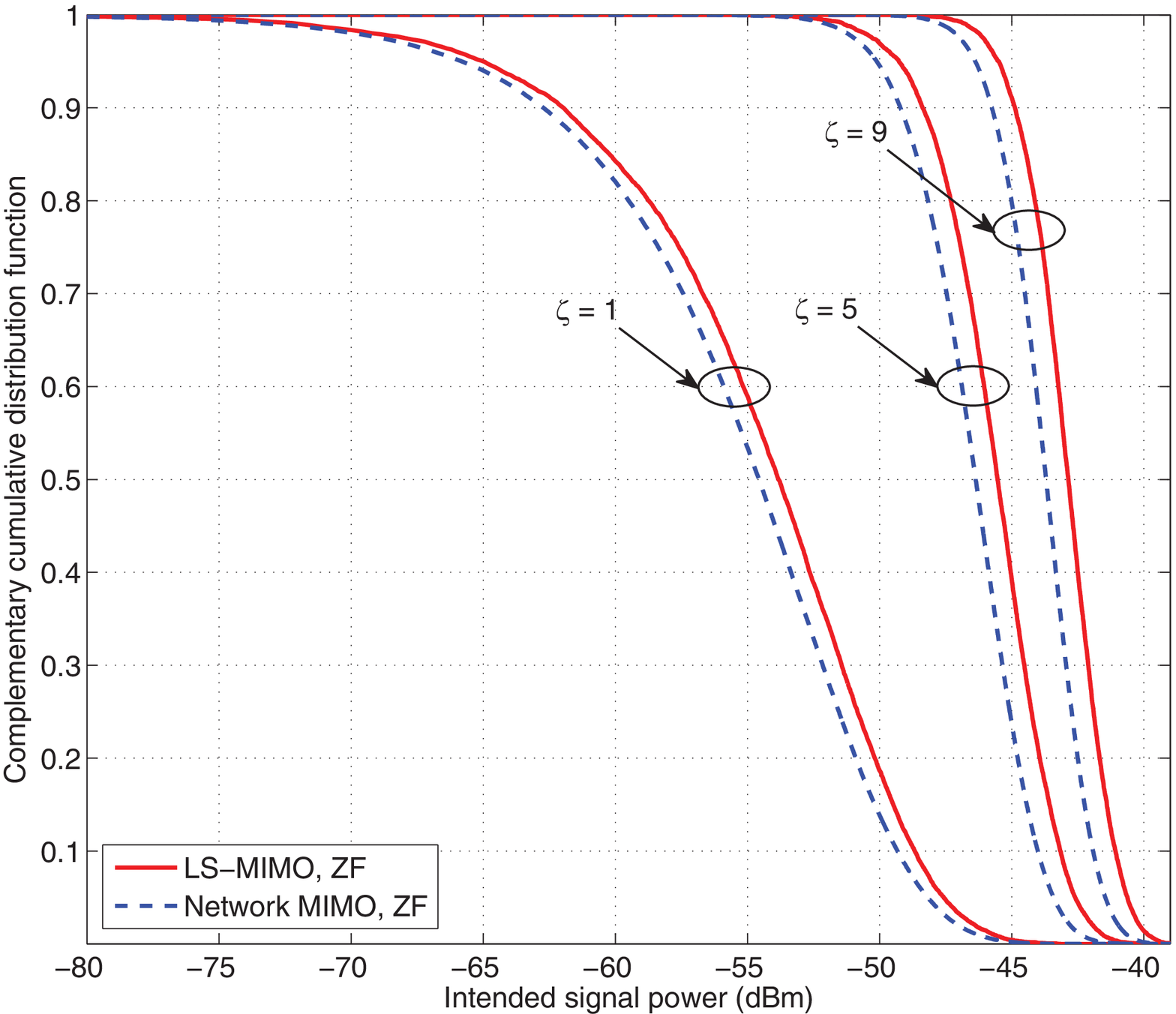}
\caption{CCDF of the signal power in the LS-MIMO and the network MIMO systems for a user located at $\LB 15 \mr{m},15  \mr{m}\RB$, almost equidistant from the cooperating BSs.}
\label{fig:eff_chan_ccdf_center}
\end{figure}

\begin{figure}[!t]
\centering
\includegraphics[scale = 0.44]{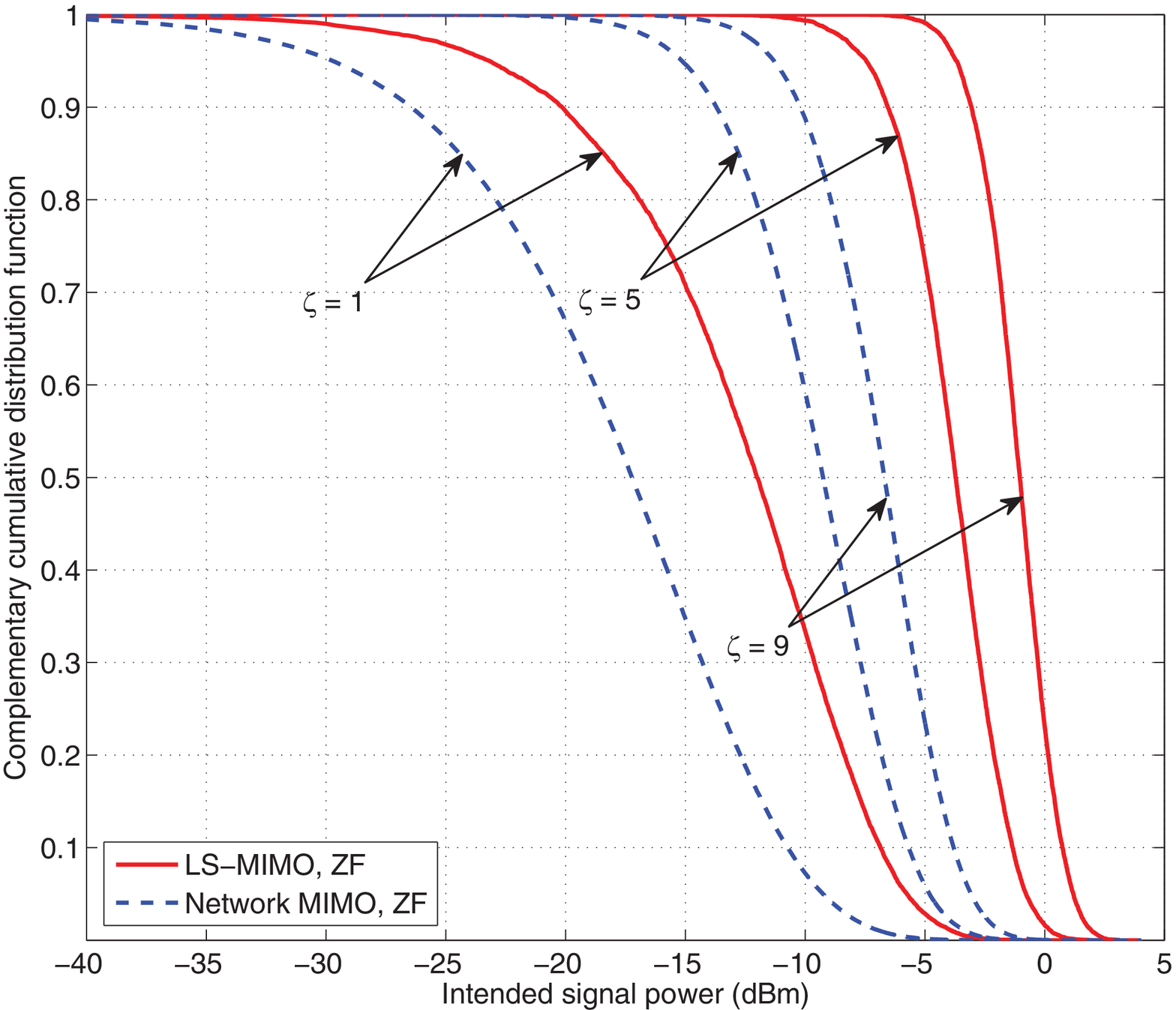}
\caption{CCDF of the signal power in the LS-MIMO and the network MIMO systems for a user located at $\LB 235 \mr{m},235  \mr{m}\RB$, closer to one of the cooperating BSs.}
\label{fig:eff_chan_ccdf_corner}
\end{figure}

$2)$ \emph{Interference power distribution}:
Fig.~\ref{fig:intf_epa} presents the cumulative distribution functions (CDF) of the inter-cluster interference powers at a fixed user location under both systems. To obtain the approximate distribution (approx. dist.) of the interference power in each system, the interference power corresponding to each single beam transmission is drawn from distribution functions presented in~\eqref{eq:ic_intf} and~\eqref{eq:nm_intf_dist} for the LS-MIMO and the network MIMO systems, respectively. The results are then averaged over the realizations of the corresponding distribution functions. Similar to the previous case, the numerical results are averaged over small-scale channel fading in the LS-MIMO system and the remaining user locations and channel realizations in the network MIMO system. As shown in the figure, the theoretical CDF curves obtained through the approximate distributions are indistinguishable. Note that due to the assumptions presented in Remark~\ref{intf_app}, the numerical CDF curves do not match perfectly with those obtained from approximate distributions. However, the gap between them is small.

\begin{figure}[!t]
\centering
\includegraphics[scale = 0.44]{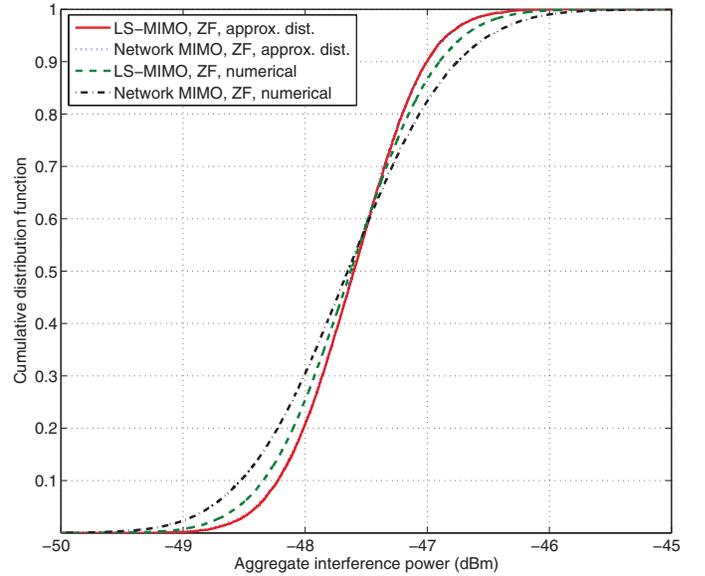}
\caption{CDF of the aggregate interference power at a user located at $\LB 15 \mr{m}, 15  \mr{m}\RB$, close to the cluster center, under both LS-MIMO and the network MIMO systems.}
\label{fig:intf_epa}
\end{figure}

$3)$ \emph{Downlink rates with round-robin scheduling}: Fig.~\ref{fig:wsr_epa} plots the users' downlink rates in the center cluster. Each BS chooses $K = 5$ out of $K_T = 60$ randomly scattered users using a round-robin scheduling scheme from within its cell area to be served during each time-slot, i.e., the scheduling weight of each user is $1/12$. The results are again averaged over both user locations and small-scale channel fading. In complete agreement with our derivations in Section~\ref{sec:epa}, the LS-MIMO system provides significant performance gains as compared to the network MIMO system. In particular, Fig.~\ref{fig:wsr_epa} shows that the rate improvement is almost $55 \%$ for the $10^{\mr{th}}$ percentile users, and this holds under different numbers of transmit antennas. It is also noticeable that even the LS-MIMO system with $\zeta = 5$ outperforms the network MIMO system with $\zeta = 9$.

$4)$ \emph{Performance evaluation with RZF}: Fig.~\ref{fig:wsr_rzf} compares the performance gains of the two systems with ZF and RZF beamforming schemes, where the achievable SINR of user $i$ in cluster $l$ of an LS-MIMO system and a network MIMO system with RZF are, respectively, given by~\eqref{eq:sinr_ic_rzf} and~\eqref{eq:sinr_nm1_rzf} at the top of the next page. Here, $\LB b,k \RB $ denotes the possible BS and user associations. Further, the regularization factor is set to $1/\rho$ for each user.
\begin{figure*}[!t]
\begin{align}
\gamma_{ibl}^{\mr{LSM,RZF}} &= \frac{\rho \lvert \g_{ilbl}\htp \w_{ibl}\rvert^2}{\sum_{\LB b',k \RB \neq \LB b,i \RB} \rho \lv \g_{ilb'l} \htp \w_{kb'l}\rv^2  +\sum_{j \neq l} \sum_{m} \sum_{k} \rho \lvert \f_{ilmj} \htp \w_{kmj} \rvert^2 + 1}\label{eq:sinr_ic_rzf}  \\
\mr{\gamma}_{il}^{\mr{NM,RZF}} &= \frac{\rho \lv \g_{il}\htp \w_{il}\rv^2}{\sum_{k \neq i}  \rho \lv \g_{il}\htp \w_{kl}\rv^2 +\sum_{ j \neq l} \sum_{k} \rho \lv \f_{ilj}\htp \w_{kj}\rv^2  + 1}. \label{eq:sinr_nm1_rzf} \\
\hline \nonumber
\end{align}
\end{figure*}
Similar to the previous case, each BS chooses $K = 5$ users based on a round-robin scheduling to be served during each time-slot. As the figure illustrates, employing RZF provides considerable gains in both systems. Moreover, even with RZF beamforming, the LS-MIMO system outperforms the comparable network MIMO system.

\begin{figure}[!t]
\centering
\includegraphics[scale = 0.44]{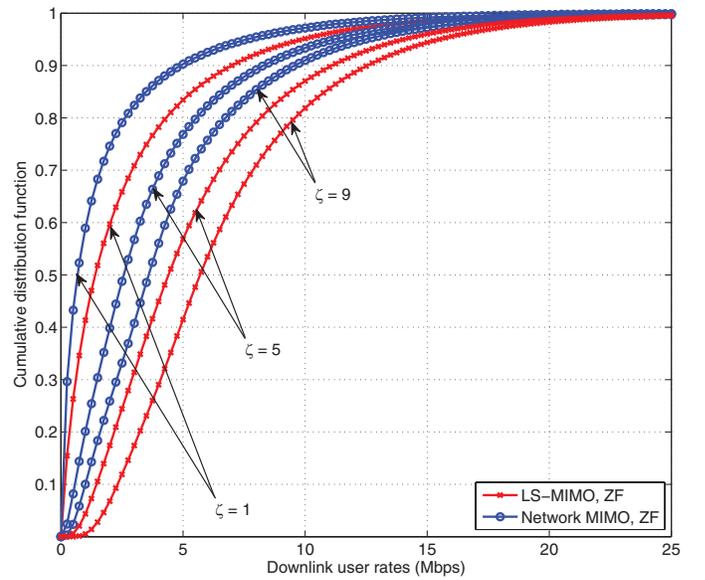}
\caption{CDF of the downlink rates using both coordination approaches with ZF.}
\label{fig:wsr_epa}
\end{figure}

\begin{figure}[!t]
\centering
\includegraphics[scale = 0.44]{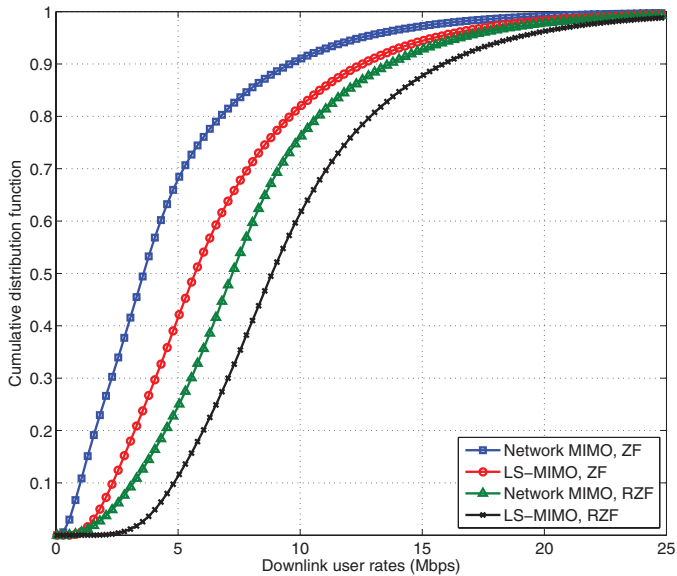}
\caption{CDF of the downlink rates using both coordination approaches with ZF and RZF, and $\zeta = 9$.}
\label{fig:wsr_rzf}
\end{figure}

\section{Concluding Remarks} \label{sec:conc}
This paper compares two distinct downlink multicell interference mitigation techniques: LS-MIMO and network MIMO. The two considered systems in this paper distribute their available power equally across the downlink beams, are capable of completely eliminating intra-cluster interference, while providing the same number of spatial DoF per user, and are subject to identical CSI acquisition overhead. Hence, it is not obvious whether one system outperforms the other. By a careful analysis of the propagation characteristics of the two systems, however, this paper shows that a network MIMO system suffers from the fact that users' intended signals are delivered by multiple scattered BSs rather than by the closest BS as in an LS-MIMO system. In particular, the disparity in the channel strengths between a user and its serving BSs introduces additional penalty in terms of the received signal power in a network MIMO system as compared to an LS-MIMO system. Further, this paper shows that, even though the two systems have distinct considerations in choosing their ZF beam directions, the aggregate inter-cluster interference power seen by each user is identically distributed under both systems. By incorporating these two results using tools from stochastic orders, we show that an LS-MIMO system provides considerable performance improvement under a wide range of utility functions as compared to a network MIMO system. Further, while the analytical comparison in this paper relies on employing ZF beamforming in each cluster, we show numerically that our conclusion also holds true in a cellular network where RZF is employed in each individual cluster. Given the likely lower cost of adding excess number of antenna elements at each BS versus joint data processing and establishing backhaul links across the BSs, the main implication of this paper is that LS-MIMO could be the preferred approach for multicell interference mitigation in wireless networks.
\appendices

\section{Proof of Theorem~\ref{sig_order}} \label{sec:app0}
\begin{IEEEproof}
Without loss of generality, we consider user $i$ located in the cell region of BS $b^*$ in cluster $l$, i.e., $b^{*}$ is its closest BS and compare the shape and scale parameters of the Gamma random variables representing the distribution of the signal power in a $\LB C,B,BM,K_c\RB$ LS-MIMO system and a $\LB C,B,M,K_c\RB$ network MIMO system.

From the shape parameter in~\eqref{eq:nm_channel}, we have that
\begin{equation}
k_{il} \leq BM \nonumber
\end{equation}
with equality only in the zero-probability case of user $i$ equidistant from its $B$ serving BSs, i.e., $\beta_{ilbl} = \beta_{ilb'l},~\forall b \neq b'$. Thus
\begin{equation}
 \frac{k_{il}\LB BM - K_c + 1\RB}{BM} \leq BM - K_c + 1 \label{eq:shape_comp}
\end{equation}
where the left-hand side and the right-hand side of the inequality~\eqref{eq:shape_comp} are the shape parameters of the Gamma random variables representing the distribution of the signal power in network MIMO and LS-MIMO systems, respectively. 

Next, considering the scale parameter in~\eqref{eq:nm_channel}, it follows that
\begin{align}
\theta_{il} &= \frac{\sum_{b = 1}^B \beta_{ilbl}^2}{\sum_{b = 1}^B \beta_{ilbl}} = \frac{\beta_{ilb^*l} \LB 1 + \sum_{b \neq b^*} \LB\frac{\beta_{ilbl}}{\beta_{ilb^*l}}\RB^2 \RB}{1 + \sum_{b \neq b^*} \LB \frac{\beta_{ilbl}}{\beta_{ilb^*l}}\RB} \ \stackrel{\LB a \RB}{\leq} \beta_{ilb^*l} \label{eq:scale_ine}
\end{align}
where relation $\LB a \RB$ follows from the fact that user $i$ is closer to BS $b^{*}$ than other BSs in cluster $l$, i.e. 
\begin{equation}
\frac{\beta_{ilbl}}{\beta_{ilb^*l}} \leq 1,~\forall b \neq b^*. \nonumber
\end{equation}
Note that the the left-hand side and the right-hand side of the inequality~\eqref{eq:scale_ine} are the scale parameters of the Gamma random variables representing the distribution of the signal power in network MIMO and LS-MIMO systems, respectively. 

Therefore, both the shape and the scale parameters are larger under an LS-MIMO system as compared to a network MIMO system. Since the CCDF of a Gamma random variable is increasing in its parameters, we conclude that the signal power of any chosen user in an LS-MIMO system first-order stochastically dominates that of a network MIMO system.
\end{IEEEproof}

\section{Proof of Theorem~\ref{epa_intf_dist}} \label{sec:app1}
\begin{IEEEproof}
In both LS-MIMO and network MIMO systems, the interference signals due to the transmissions initiated by different clusters are independent. Thus, without loss of generality, we only focus on the interference power distribution caused by cluster $j$ at user location $i$ within cluster $l$.

The aggregate interference power created by cluster $j$ in an LS-MIMO system is given by
\begin{equation}
\sum_{m = 1}^{B}\sum_{k = 1}^{K}  \lv \f_{ilmj}\htp \w_{kmj}\rv^2 \label{eq:ic_intf_one_cluster}
\end{equation}
where the $m^{\mr{th}}$ term (the inner sum in~\eqref{eq:ic_intf_one_cluster}) denotes the interference power produced by BS $m$ in cluster $j$. The interference produced by BS $m$ is itself a sum of $K$ terms, each associated with the transmission of a single ZF beam from BS $m$. In this paper, we assume that the $K$ ZF beams designed by each interfering BS of an LS-MIMO system are orthogonal (Assumption A$2$) . Therefore, the interference power imposed by BS $m$ in cluster $j$ is the power of $ \f_{ilmj}$ projected onto $K$ orthogonal subspaces. This is equivalent to a summation of $K$ independent Gamma random variables; each term is distributed as $\Gamma \LB 1,\beta_{ilmj}\RB$ as presented in~\eqref{eq:ic_intf}. Noting that these $K$ Gamma random variables have identical scale parameters, the aggregate interference power from BS $m$ in cluster $j$ is distributed as 
\begin{equation}
\sum_{k = 1}^{K}  \lv \f_{ilmj}\htp \w_{kmj}\rv^2  \sim \Gamma \LB K,\beta_{ilmj}\RB. \label{eq:intf_clus_ls}
\end{equation}

Next, we note that the interference signals produced by different BSs of cluster $j$ are independent. Specifically, since in an LS-MIMO system the ZF beams designed at each BS are only dependent on the small-scale channel fading between the BS and the set of $K_c$ users in its associated cluster, and the small-scale channel fading is independent across the BSs, the summation in~\eqref{eq:ic_intf_one_cluster} is a sum of $B$ independent, but non-identically distributed Gamma random variables with the $m^{\mr{th}}$ term distributed according to~\eqref{eq:intf_clus_ls}. Therefore, using Approximation~\ref{app_1}, we have
\begin{multline}
\sum_{m = 1}^{B}\sum_{k = 1}^{K}  \lv \f_{ilmj}\htp \w_{kmj}\rv^2 \sim \\
 \Gamma \LB K\frac{\LB\sum_{m = 1}^B \beta_{ilmj} \RB^2}{\sum_{m = 1}^B \beta_{ilmj}^2},  \frac{\sum_{m = 1}^B \beta_{ilmj}^2}{\sum_{m = 1}^B \beta_{ilmj}} \RB. \label{eq:ic_intf_one_cluster1}
\end{multline}

Next, we obtain the distribution of the total interference power produced by cluster $j$ in a network MIMO system which is a summation of $K_c$, not necessarily independent, terms given by
\begin{equation}
\sum_{k = 1}^{K_c}  \lv \f_{ilj}\htp \w_{kj}\rv^2 \label{eq:sum_intf_power_nm}
\end{equation}
where, based on Lemma~\ref{nm_intf_power_dist}, each term is distributed as in~\eqref{eq:nm_intf_dist}. Here, we assume that the $K_c$ ZF beams designed at each interfering cluster of a network MIMO system are orthogonal (Assumption A2). As a consequence, the summation in~\eqref{eq:sum_intf_power_nm} is a sum of $K_c$ independent Gamma random variables which have the same scale parameters. Therefore, it follows that
\begin{align}
&\sum_{k = 1}^{K_c}  \lv \f_{ilj}\htp \w_{kj}\rv^2 \sim \Gamma \LB K\frac{\LB\sum_{m = 1}^B \beta_{ilmj} \RB^2}{\sum_{m = 1}^B \beta_{ilmj}^2}, \frac{\sum_{m = 1}^B \beta_{ilmj}^2}{\sum_{m = 1}^B \beta_{ilmj}} \RB. \label{eq:nm_intf_one_cluster}
\end{align}

From~\eqref{eq:ic_intf_one_cluster1} and~\eqref{eq:nm_intf_one_cluster}, the total interference powers produced by cluster $j$ at user $i$ are identically distributed. Given that the received interference signals from different clusters are independent, it follows that the aggregate interference power at any chosen user location under the two systems are equal in distribution. This completes the proof.
\end{IEEEproof}

\section{Proof of Theorem~\ref{sinr_order}}\label{sec:app2}
\begin{IEEEproof}
Without loss of generality, we evaluate the achievable SINR of user $i$ in cluster $l$ under both systems as given by~\eqref{eq:sinr_ic} and~\eqref{eq:sinr_nm1}.


Recalling that $I_{il}^{\mr{LSM}}$ and $I_{il}^{\mr{NM}}$, respectively, denote the aggregate interference power seen by user $i$ in cluster $l$ in LS-MIMO and network MIMO systems, based on Theorem~\ref{epa_intf_dist}, we have that $I_{il}^{\mr{LSM}} \stackrel{d}{=}I_{il}^{\mr{NM}}$. Hence, it follows that
\[\frac{\left(\rho I_{il}^{\mr{LSM}} + 1\right)}{\rho} \stackrel{d}{=}\frac{\left(\rho I_{il}^{\mr{NM}} + 1\right)}{\rho}.\]

For convenience, let $\bar{\gamma} \stackrel{d}{=} \left(\rho I_{il}^{\mr{LSM}} + 1\right)/\rho \stackrel{d}{=} \left(\rho I_{il}^{\mr{NM}} + 1\right)/\rho$ and $p_{\bar{\gamma}}(\cdot)$ denote the common distribution of these two random variables. Further, let $X_{ibl} = \lvert \g_{ilbl}\htp \w_{ibl}\rvert^2$ and $Y_{il} = \lv \g_{il}\htp \w_{il}\rv^2$. Therefore
\begin{eqnarray}
\PP \LCB \gamma_{ibl}^{\mr{LSM,ZF}} \geq \gamma_0\RCB & = &  \PP \LCB X_{ibl} \geq \gamma_0 \frac{\left(\rho I_{il}^{\mr{LSM}} + 1\right)}{\rho}\RCB \nonumber \\
 & = & \int_0^\infty \PP\LCB X_{ibl} \geq \gamma_0 \bar{\gamma} \lvert \bar{\gamma} \RCB p_{\bar{\gamma}}(\bar{\gamma}) d\bar{\gamma} \nonumber \\
 & \geq & \int_0^\infty \PP\LCB Y_{il} \geq \gamma_0 \bar{\gamma} \lvert \bar{\gamma} \RCB p_{\bar{\gamma}}(\bar{\gamma}) d\bar{\gamma} \nonumber \\
 & = & \PP\LCB \gamma_{il}^{\mr{NM,ZF}} \geq \gamma_0 \RCB \nonumber
\end{eqnarray}
where the inequality follows from Theorem 1. Since this result holds for every choice of $\gamma_0$, it implies that
\begin{equation}
\gamma_{ibl}^{\mr{LSM,ZF}} \geq_{st} \gamma_{il}^{\mr{NM,ZF}}. \nonumber
\end{equation}
%
\end{IEEEproof}

\bibliographystyle{IEEEtran}
\bibliography{IEEEabrv,MyRef1}


\begin{IEEEbiography}[{\includegraphics[width=1in,height=1.25in,clip,keepaspectratio]{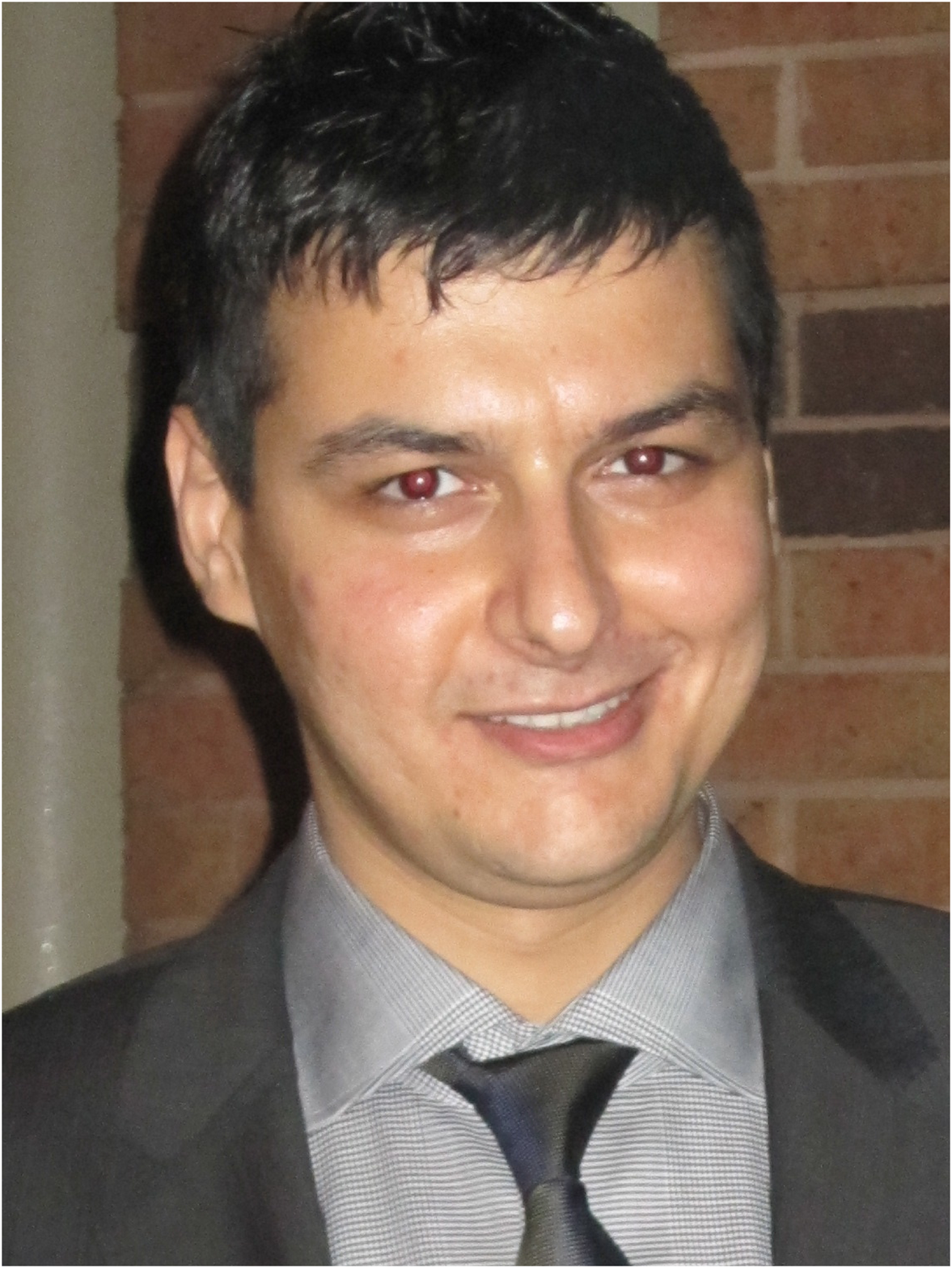}}]{Kianoush Hosseini}
(S'10) was born in Tehran, Iran. He received his B.Sc. degree in Electrical Engineering from Iran University of Science and Technology, Tehran, Iran in 2008. Since 2008, he has been with The Edward S. Rogers Sr. Department of Electrical and Computer Engineering at University of Toronto, Toronto, Ontario, Canada where he received his M.A.Sc. degree in 2010, and is currently pursuing the Ph.D. degree. During Summer 2012 and 2014, he was an intern at Bell-Labs, Alcatel-Lucent, in Stuttgart, Germany, and Qualcomm Corporate Research and Development, in San Diego, USA, respectively. Mr. Hosseini is a recipient of the Edward S. Rogers Sr. Graduate Scholarship. His main research interests include MIMO communications, optimization theory, information theory, and distributed algorithms.
\end{IEEEbiography}

\begin{IEEEbiography}[{\includegraphics[width=1in,height=1.25in,clip,keepaspectratio]{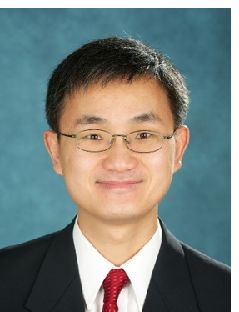}}]{Wei Yu}
(S'97-M'02-SM'08-FÕ14) received the B.A.Sc. degree in Computer Engineering and Mathematics from the University of Waterloo, Waterloo, Ontario, Canada in 1997 and M.S. and Ph.D. degrees in Electrical Engineering from Stanford University, Stanford, CA, in 1998 and 2002, respectively. Since 2002, he has been with the Electrical and Computer Engineering Department at the University of Toronto, Toronto, Ontario, Canada, where he is now Professor and holds a Canada Research Chair in Information Theory and Wireless Communications. His main research interests include information theory, optimization, wireless communications and broadband access networks.

Prof. Wei Yu served as an Associate Editor for IEEE Transactions on Information Theory (2010-2013), as an Editor for IEEE Transactions on Communications (2009-2011), as an Editor for IEEE Transactions on Wireless Communications (2004-2007), and as a Guest Editor for a number of special issues for the IEEE Journal on Selected Areas in Communications and the EURASIP Journal on Applied Signal Processing. He was a Technical Program Committee (TPC) co-chair of the Communication Theory Symposium at the IEEE International Conference on Communications (ICC) in 2012, and a TPC co-chair of the IEEE Communication Theory Workshop in 2014. He was a member of the Signal Processing for Communications and Networking Technical Committee of the IEEE Signal Processing Society (2008-2013). Prof. Wei Yu received an IEEE ICC Best Paper Award in 2013, an IEEE Signal Processing Society Best Paper Award in 2008, the McCharles Prize for Early Career Research Distinction in 2008, the Early Career Teaching Award from the Faculty of Applied Science and Engineering, University of Toronto in 2007, and an Early Researcher Award from Ontario in 2006.

Prof. Wei Yu is a Fellow of IEEE. He is a registered Professional Engineer in Ontario.
\end{IEEEbiography}

\begin{IEEEbiography}[{\includegraphics[width=1in,height=1.25in,clip,keepaspectratio]{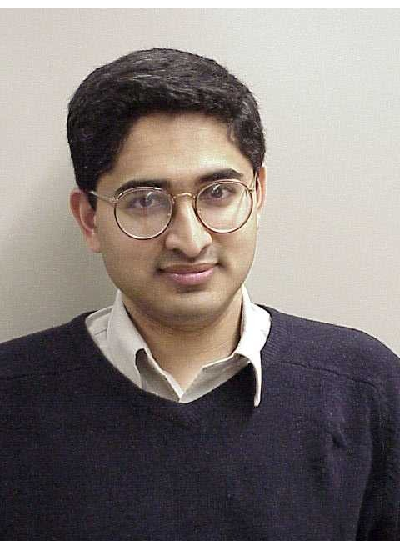}}]{Raviraj S. Adve}
(S'88-M'97-SM'06) was born in Bombay, India. He received
his B. Tech. in Electrical Engineering from IIT, Bombay, in 1990 and his
Ph.D. from Syracuse University in 1996. Between 1997 and August 2000, he
worked for Research Associates for Defense Conversion Inc. on contract with
the Air Force Research Laboratory at Rome, NY. He joined the faculty at the
University of Toronto in August 2000 where he is currently a Professor. Dr.
Adve's research interests include analysis and design techniques for
heterogeneous networks, energy harvesting networks and in signal processing
techniques for radar and sonar systems. He received the 2009 Fred Nathanson
Young Radar Engineer of the Year award.
\end{IEEEbiography}

\end{document}